\documentclass{biometrika}

\usepackage{caption}
\usepackage{amsmath}
\makeatletter
\providecommand\@LN[1]{}  
\makeatother

\usepackage{tikz}
\usetikzlibrary{shapes.geometric}
\usepackage{graphicx}
\usepackage{amsmath, amsfonts, amssymb}
\usepackage{multirow}
\usepackage{subcaption}
\usepackage{comment}

\usepackage{longtable}
\usepackage{algorithm}
\usepackage{algorithmic} 
\newcommand{\indep}{\rotatebox[origin=c]{90}{$\models$}}

\usepackage{booktabs}
\usepackage{needspace}
\usepackage[colorlinks=true,
            linkcolor=blue,
            citecolor=blue,
            urlcolor=blue]{hyperref}

\usepackage{newtxtext}
\usepackage[subscriptcorrection]{newtxmath}

\graphicspath{{./art/}}



\begin{document}

\jname{arXiv Preprint}


\markboth{P. Heng et~al.}{Biometrika}

\title{Revisiting Madigan and Mosurski: Collapsibility via Minimal Separators}

\author{Pei Heng}
\affil{KLAS and School of Mathematics and Statistics, Northeast Normal University,\\ Changchun, China
	\email{peiheng@nenu.edu.cn}}

\author{Shiyuan He}
\affil{The Academy for Advanced Interdisciplinary Studies, Beijing Key Laboratory of Applied Statistics and Digital Regulation, School of Mathematics and
	Statistics, Beijing Technology and Business University,\\ Beijing, China
	\email{heshiyuan@btbu.edu.cn}}

\author{Yi Sun}
\affil{Institute of Statistics and Data Science, Xinjiang University of Finance and Economics,\\ Urumqi, China \email{brian@xjufe.edu.cn}}

\author{and Jianhua Guo}
\affil{The Academy for Advanced Interdisciplinary Studies, Beijing Key Laboratory of Applied Statistics and Digital Regulation, School of Mathematics and
	Statistics, Beijing Technology and Business University,\\ Beijing, China
	\email{jhguo@btbu.edu.cn}}

\maketitle

\begin{abstract}
	Collapsibility provides a principled approach for dimension reduction in contingency tables and graphical models. \citet{madigan1990} pioneered the study of minimal collapsible sets in decomposable models, but existing algorithms for general graphs remain computationally demanding. We show that a model is collapsible onto a target set precisely when that set contains  {at least one minimal separator} between its non-adjacent vertices. This insight motivates the Close Minimal Separator Absorption (CMSA) algorithm, which constructs minimal collapsible sets using only local separator searches at very low costs. Simulations confirm substantial efficiency gains, making collapsibility analysis practical in high-dimensional settings.

\end{abstract}

\vspace{-27pt}

\section{Introduction}

Collapsibility is a classical idea in multivariate statistical analysis. First introduced by \cite{yule1903notes} and \cite{simpson1951interpretation}, it provides a principled way to simplify statistical analysis by removing variables without distorting marginal associations. Within the framework of log-linear models, \cite{asmussen1983}, in a seminal {\it Biometrika} article, gave the first rigorous account, showing that collapsibility guarantees that marginal inferences coincide with those drawn from the full model. This property makes collapsibility an attractive tool for dimension reduction in contingency tables, graphical models, and related methods.

Also in {\it Biometrika}, a sharper problem was posed by \cite{madigan1990}:  for a target set of variables, what is the smallest superset onto which the model can be collapsed without loss of inferential validity? Such a superset is called a minimal collapsible set.
\cite{madigan1990} proposed the selective acyclic hypergraph reduction (SAHR) algorithm for searching the minimal collapsible set. Their algorithm is effective but only works for decomposable graphical models. Later work sought to extend this idea to more general graphical models: \cite{wgh2011} developed a convex-hull based procedure, while \cite{heng2023}  introduced path-absorption methods. These approaches, however, often require global graph operations and become computationally expensive for high-dimensional models.

In this note, we revisit minimal collapsibility from a new perspective. We show that a graphical model is collapsible onto a target set if and only if that set contains  {at least one minimal separator} between non-adjacent vertices in the set. This characterization, simple yet powerful, shifts focus from global graph operations to iterative separation structure searching, which is localized with several connected components with decreasing size. Based on this observation, we introduce the Close Minimal Separator Absorption (CMSA) algorithm, which iteratively absorbs close minimal separators within the neighbours of the target vertices at very low cost.


The benefits of CMSA are twofold. Conceptually, it reveals that collapsibility is governed by purely local graph properties. Computationally, it yields dramatic improvements in efficiency: in decomposable models, CMSA consistently outpaces SAHR, while in general graphical models, it scales far better than the existing algorithms.
A simple worked example illustrates the mechanism, and simulations on large random graphs confirm its speed and robustness. Our contribution is therefore both theoretical and practical: we provide a clean separator-based characterization of collapsibility, and a localized algorithm that makes collapsibility analysis feasible in high-dimensional graphical models.

\vspace{-12pt}
\section{Hierarchical log-linear model and Graph Representation}\label{sec2}

\cite{asmussen1983} analysed hierarchical log-linear models for multidimensional contingency table $N$, where each cell  $i$ is determined by a set $V$ of classifying factors.  Specifically, each factor in  $V$ is a categorical variable for classifying observations.  
Let $p(i)$ be the probability of an observation falling into a cell $i$. 
A hierarchical log-linear model $\mathcal{L}$ consists of a class of probabilities ($p\in \mathcal{L}$)  characterized by its generating class. The generating class has generators represented with square brackets. For example,  for a 4-way table with $V=\{\alpha, \beta, \gamma, \theta\}$, the log-linear model with generator $[\alpha\beta][\beta\gamma\theta]$ is, 
\begin{equation} \label{eqn:firstexample}
	\log (m_{hjkl}) = \lambda+\lambda_h^\alpha+\lambda_j^\beta+\lambda_k^\gamma+\lambda_l^\theta+\lambda_{hj}^{\alpha\beta}+\lambda_{jk}^{\beta\gamma}+\lambda_{k l}^{\gamma\theta}+\lambda_{j l}^{\beta\theta}+\lambda_{j k l}^{\beta\gamma\theta}.
\end{equation}
In the above, $m_{ijkl}$ represents the expected number of observations in the cell, whose classifying factors $(\alpha, \beta, \gamma, \theta)$ equal $(h,j,k,l)$. Besides, $\lambda_h^\alpha$ denotes the marginal effect of the variable $\alpha$,  $\lambda_{h j}^{\alpha\beta}$ represents the first-order interaction effect between $\alpha$ and $\beta$, and so on.

Log-linear models can be analysed from the graphical perspective \citep{asmussen1983}.  
An undirected graph $G=(V,E)$ consists of a vertex set $V$ and a edge set  $E\subseteq 
V\times V=\{(x,y):x,y\in V\}$. Two vertices are connected by an edge if they are in some generator.  For the  model $\mathcal{L}$ in Eqn.~\eqref{eqn:firstexample} and the model $\mathcal{L}^\prime = [\alpha\beta][\beta\gamma][\gamma\theta][\beta\theta]$, they share the same interaction graph  in Fig.~\ref{Fig1} (a) with the vertex set  $V = \{\alpha, \beta, \theta, \gamma\}$. \emph{Graphical models} are statistical models that satisfy the conditional independence relations implied by the Markov property with respect to a given graph. Note that $\mathcal{L}$ is a graphical model, whereas $\mathcal{L}^\prime$ is not.  


\begin{figure}[h]
	\centering	
	\begin{subfigure}[b]{0.3\textwidth}
		\centering
		\begin{tikzpicture}[
			scale=0.45, 
			transform shape,
			every node/.style={circle, draw=black, minimum size=7mm, inner sep=0pt, font=\large}
			]
			\node (a) at (-1.5,-1.2) {$\alpha$};
			\node (b) at (0,-1.2)    {$\beta$};
			\node (c) at (2,-1.2)    {$\gamma$};
			\node (d) at (1.1,-2.4)  {$\theta$};
			\draw[thick] (b)--(a);
			\draw[thick] (b)--(c);
			\draw[thick] (d)--(c);
			\draw[thick] (b)--(d);
		\end{tikzpicture}
		\caption{} 
	\end{subfigure}
	\hfill 
	\begin{subfigure}[b]{0.3\textwidth}
		\centering
		\begin{tikzpicture}[
			scale=0.45,
			transform shape,
			every node/.style={circle, draw=black, minimum size=7mm, inner sep=0pt, font=\large}
			]
			\node (1) at (0,0){$a$};
			\node (2) at (0,-1){$t$};
			\node (3) at (1.8,-1){$l$};
			\node [fill=gray!30] (4) at (0.8,-2){$e$};
			\node (5) at (0,-3){$x$};
			\node (6) at (2.4,-3){$d$};
			\node (7) at (3,-1.3){$b$};
			\node [fill=gray!30] (8) at (2.4,0){$s$};
			\draw[-,thick] (2)--(1); \draw[-,thick] (2)--(3); \draw[-,thick] (4)--(3);
			\draw[-,thick] (2)--(4); \draw[-,thick] (4)--(5); \draw[-,thick] (4)--(6);
			\draw[-,thick] (4)--(7); \draw[-,thick] (8)--(3); \draw[-,thick] (8)--(7);
			\draw[-,thick] (6)--(7);
		\end{tikzpicture}
		\caption{} 
	\end{subfigure}
	\hfill
	\begin{subfigure}[b]{0.3\textwidth}
		\centering
		\begin{tikzpicture}[
			scale=0.45,
			transform shape,
			every node/.style={circle, draw=black, minimum size=7mm, inner sep=0pt, font=\large}
			]
			\node (3) at (1.8,-1){$l$};
			\node [fill=gray!30] (4) at (0.8,-2){$e$};
			\node (7) at (3,-2){$b$};
			\node [fill=gray!30] (8) at (2.4,0){$s$};
			\draw[-,thick] (4)--(3); \draw[-,thick] (4)--(7);
			\draw[-,thick] (8)--(3); \draw[-,thick] (8)--(7);
		\end{tikzpicture}
		\caption{} 
	\end{subfigure}
	\hfill 
	\caption{(a) The interaction graph of $[\alpha\beta][\beta\gamma\theta]$ and $[\alpha\beta][\beta\gamma][\gamma\theta][\beta\theta]$. (b) The moral graph of the Asia network \citep{lauritzen1988local}. (c) The minimal collapsible subgraph $G_B$ containing $\{e,s\}$. }.
	\label{Fig1}
	\vspace{-0.8cm}
\end{figure}

\begin{remark}
	As in \cite{asmussen1983, madigan1990}, we let $\mathcal{L}$ be a hierarchical log-linear model. However, it is worth noting that the proposed algorithm is equally applicable to general multinomial and Gaussian graphical models. For mixed graphical models with discrete and continuous variables, our algorithm also identifies minimal collapsible sets by applying it to the constructed star graph \citep[see][for details]{Frydenberg1990,wgh2011}.
\end{remark}

We lastly review some graph terminologies. In $G$, two vertices $x$ and $y$ are adjacent if connected by an edge; then, $x$ is a neighbour of $y$, and the set of all neighbours of $x$ is denoted as $N_G(x)$. For a subset $A \subseteq V$, its neighbour set is $N_G(A) = \bigcup_{v \in A} N_G(v) \setminus A$.  A path connecting $u$ and $v$ in $G$, denoted by $l_{uv}$, is a sequence of distinct vertices and edges $(u=v_0,e_0,v_1,e_1,\ldots,v_{k-1},e_{k-1},v_k=v)$ such that $e_i=(v_i,v_{i+1})\in E$. We set $V(l_{uv})=\{v_0, v_1, \ldots, v_k\}$ to be the set of vertices on $l_{uv}$. The subgraph induced by $A$ is $G_A = (A, E_A)$ with $E_A = E \cap \{(x,y) : x,y \in A\}$. We refer to a maximal connected subgraph of $G$ as a connected component. A chordal graph is an undirected graph in which every cycle of length four or more contains a chord.

\vspace{-10pt}

\section{ Minimal collapsible set identifying algorithms}\label{sec3}

\subsection{Minimal collapsible set}

Recall, for a contingency table $N$,  $p(i)$ is the probability for an observation falling into a cell $i$, and denote $\hat{p}(i)$ as its maximum likelihood estimate (MLE), based on samples following a multinomial distribution. For \(A \subseteq V\), let $\hat{p}(i_A)$ be the probability obtained by marginalizing (summing)  $\hat{p}(i)$ over the remaining factors $V\setminus A$. This marginalized probability may change the independence structure as specified by $\mathcal{L}$ between variables in $A$. Alternatively, to keep the independence structure, we can set $\mathcal{L}_A$ as the marginal model obtained by removing all factors in $V\setminus A$ and then deleting all redundant generators (those are contained within another remaining generator). For example, for the model $\mathcal{L} = [\alpha\beta][\beta\gamma\theta]$ in Equation~\eqref{eqn:firstexample}, assuming $A = \{\beta, \theta\}$, the marginal model is specified by $\mathcal{L}_A = [\beta][\beta\theta]= [\beta\theta]$. Let $p_{\mathcal{L}_A}$ be a probability specified by $\mathcal{L}_A$ , with MLE \(\hat p_{\mathcal{L}_A}\). 
\cite{asmussen1983} formalized the notion of \emph{collapsibility}.
In particular, a hierarchical log-linear model \(\mathcal{L}\) is  said to be \emph{collapsible onto} \(A \subseteq V\) if the two MLEs coincide:   $\hat{p}(i_A) = \hat{p}_{\mathcal{L}_A}(i_A)$.
Collapsibility ensures that marginal inferences are consistent with those drawn from the full model. In practice, such property can yield substantial savings in data collection and computation efforts, while enhancing robustness to unobserved variables.

Given a decomposable graphical model $\mathcal{L}$ and a vertex subset $A \subseteq V$,  \cite{madigan1990} introduced the problem of identifying a minimal collapsible set: what is the smallest set $B$ with $A \subseteq B \subseteq V$ such that $\mathcal{L}$ is collapsible onto $B$? They proposed a simple and efficient procedure for finding such a set by iteratively removing simplicial vertices not in $A$ until no further removal is possible. The existence and uniqueness of the minimal collapsible set for general graphical models were later established in Theorem~2.4 of \cite{wgh2011}.


\subsection{Collapsibility characterization via separability}

The approach of \cite{madigan1990} only works for decomposable graphical models. For general graphical models, we present a novel characterization of collapsibility based on pairwise separability on a graph. For pairwise disjoint subsets $A, B, S \subseteq V$ in $G$, if $S \cap V(l_{ab}) \neq \emptyset$ for every path $l_{ab}$ connecting any $a \in A$ and $b \in B$, then $S$ separates $A$ from $B$ in $G$, denoted by $A \indep B \mid S[G]$. The set $S$ is called an $AB$-separator in $G$, and it is minimal if no proper subset of $S$ separates $A$ from $B$ in $G$. When $A = \{x\}$ and $B = \{y\}$, $S$ is also called a minimal $xy$-separator.
The proposed method is motivated by the following key lemma.

\begin{lemma}[\citet{asmussen1983}]\label{prop-3}
	The graphical model $\mathcal{L}$ is collapsible onto subset $A\subseteq V$ if and only if $X\indep Y\mid Z[G]$ implies $X\indep Y\mid Z\cap A [G]$ for any $X,Y\subseteq A$.
\end{lemma}

Lemma \ref{prop-3} implies that for a collapsible subset $A$, any minimal $xy$-separator $Z$ for non-adjacent vertices $x, y \in A$ is contained within $A$. Otherwise, $ Z\cap A$ would form a smaller separator, contradicting the minimality of $Z$. This observation motivates the following theorem, which formally characterizes the relationship between collapsibility and minimal separators.

\begin{theorem}\label{thm-0}
	{The graphical model $\mathcal{L}$ is collapsible onto a subset $A \subseteq V$ if and only if $A$ contains at least one minimal $xy$-separator for every pair of non-adjacent vertices $x, y \in A$.}
\end{theorem}

\begin{proof}
	{The necessity follows by taking $Z$ as a minimal separator in Lemma~\ref{prop-3}. For sufficiency, assume that $\mathcal{L}$ is not collapsible onto $A$. Then, by Theorem 2.3 of \cite{asmussen1983}, there exists a connected component $M$ of $G_{V \backslash A}$ whose neighbour set $N_G(M) \subseteq A$ contains a non-adjacent pair $x, y$. This implies that any subset of $A \backslash\{x, y\}$ fails to separate $x$ and $y$, since there exists a path connecting $x$ and $y$ entirely through $M$. This contradicts the condition that $A$ contains at least one minimal $x y$-separator for every non-adjacent pair in $A$, completing the proof.}
\end{proof}
{
	According to Theorem~\ref{thm-0}, a minimal collapsible set containing \(A\) can be constructed by iteratively absorbing minimal separators associated with non-adjacent vertex pairs in \(A\) until each non-adjacent pair in the updated set contains at least one minimal separator. In the next subsection, we show that this procedure can be carried out efficiently through localized searches.
}

\subsection{Close minimal separator absorption  algorithm}


{Building on Theorem~\ref{thm-0}, the iterative procedure can be implemented by absorbing a single minimal separator for each pair of non-adjacent vertices in the current set.} More importantly, the minimal separator required for each pair is what is known as a \emph{close minimal separator}, defined as one that lies entirely within the neighbourhood of a vertex. The importance of this notion is that it reduces the problem to localized searches.

\begin{definition}[{\cite{Takata2010}}]
	For any two non-adjacent vertices $x, y \in V$, a minimal $xy$-separator $S$ is said to be close to $x$ if $S \subseteq N_G(x)$. We denote such a separator as $S^{x}_{xy}$. Similarly, a minimal $xy$-separator that is close to $y$ is denoted as $S^{y}_{xy}$.
\end{definition}


For example, in Fig.~\ref{Fig1} (b), $\{e,l\}$ is the minimal $bt$-separator close to $t$, whereas $\{e,s\}$ is the minimal $bt$-separator close to $b$. \cite{Takata2010} also proposed the  \textbf{CloseSeparator} algorithm, which identifies close minimal separators with a time complexity of $O(m)$, where $m$ denotes the number of edges in the graph $G$. The algorithm first finds the neighbour set $N_G(x)$ of $x$, and then identifies the connected component $M$ containing $y$ in the subgraph $G_{V \backslash N_G(x)}$. Since $N_G(M)\subseteq N_G(x)$, the neighbour set $N_G(M)$ forms a minimal $xy$-separator that is close to $x$.

We now formally introduce the \emph{Close Minimal Separator Absorption (CMSA)} algorithm (Algorithm~\ref{alg-1}) for efficiently identifying minimal collapsible sets. Given a hierarchical log-linear model $\mathcal{L}$ with interaction graph $G=(V, E)$ and a target set $A$, CMSA proceeds as follows:

\begin{quote}
	\begin{enumerate}
		\item[(i)] Identify all connected components $M_1, \dots, M_K$ of $G_{V\setminus A}$.
		\item[(ii)] For each connected component $M_i$, initialize $B_i := A$ and consider the subgraph $G_i = G_{B_i \cup M_i}$. Iteratively identify non-adjacent vertex pairs within the neighbourhoods of the connected components of $G_{M_i}$, and absorb their close minimal separators in $G_i$ into $B_i$. The iteration stops when, for all connected components in the updated $G_{M_i}$, their neighbours in $G_i$ form complete subsets. At this point, according to Theorem~2.3 in \cite{asmussen1983}, $G_i$ can be collapsed onto $B_i$, and therefore the iteration is terminated.
		\item[(iii)] Merge the sets $B_i$ obtained from all connected components to form the final minimal collapsible set $B = \bigcup_i B_i$, which contains $A$.
	\end{enumerate}    
\end{quote}

\vspace{-11pt}
\begin{algorithm}[h]
	\caption{Close Minimal Separator Absorption Algorithm (CMSA)}
	\label{alg-1}
	\begin{algorithmic}[1]
		\REQUIRE A graphical model $\mathcal{L}$ with its interaction graph $G=(V, E)$ and a target variable set $A$.
		\ENSURE The minimal collapsible set $B$ containing the subset $A$.
		\STATE Identify all connected components $M_1,M_2,\ldots, M_K$ of 
		$G_{V\setminus A}$;
		\FOR{ each connected component $G_{M_i}$ of $G_{V\setminus A}$}
		\STATE Initialize: $B_i:=A$, $G_i = G_{B_i\cup M_i}$;
		\REPEAT	
		\STATE If there exists a connected component $C$ of $G_{M_i}$ whose neighbour in $G_i$ contains a pair of non-adjacent vertices, select any such pair $\{u, v\} \subseteq N_{G_i}(C) \subseteq B_i$;
		\STATE Find $S^u_{uv}$ and $S^v_{uv}$ in $G_i$ using \textbf{CloseSeparator};
		\STATE Update $B_i := B_i \cup S^u_{uv} \cup S^v_{uv}$ and $M_i := M_i \setminus (S^u_{uv} \cup S^v_{uv})$;
		\UNTIL{In $G_i$, the neighbour of every connected component of  $G_{M_i}$ is complete.}
		\ENDFOR	
		\RETURN $\bigcup_{i} B_i$, where $B_i$ is the set obtained in the $i$-th connected component $M_i$.
	\end{algorithmic}
\end{algorithm}
\vspace{-5pt}

Algorithm~\ref{alg-1} searches and absorbs separators localized within several connected components with decreasing size. The following example illustrates its execution. Another example comparing CMSA and SAHR is provided in the Supplementary Material.

\begin{example}	
	Let $G = (V, E)$ be the graph shown in Fig.~\ref{Fig1}(b), and let the target variable set be $A = \{e, s\}$. At the start of the CMSA algorithm, the subgraph $G_{V \setminus A}$ consists of three connected components: $M_1 = \{x\}$, $M_2 = \{b, d\}$, and $M_3 = \{a, l, t\}$, with $B_i$ initialized as $A$ for $i = 1, 2, 3$. The algorithm then processes each connected component as follows:
	\vspace{-5pt}
	\begin{quote}
		\begin{enumerate}
			\item  For $M_1 = \{x\}$, the boundary $N_{G_{\{e,s,x\}}}(\{x\}) = \{e\}$. There are no non-adjacent vertices in $N_G(\{x\})$, so no update of $B_1$ is needed.
			\item For $M_2 = \{b, d\}$ with $G_2 = B_2 \cup M_2 = \{b, d, e, s\}$, the boundary $N_{G_2}(M_2) = \{e, s\}$ contains the non-adjacent pair $\{e, s\}$. In $G_2$, the close minimal $es$-separators are $\{b\}$ for both vertices; thus, we update $B_2 = \{b,e,s\}$ and $M_2 = \{d\}$. Now, $N_{G_2}(M_2) = \{b, e\}$, whose vertices are adjacent, so processing for this component terminates.
			
			
			\item For $M_3 = \{a, l, t\}$, the boundary also contains the non-adjacent pair $\{e, s\}$. Both close minimal $es$-separators are $\{l\}$; hence, we update $B_3 = \{e,l,s\}$ and $M_3 = \{a,t\}$. No further absorption is required for this component.
		\end{enumerate}    
	\end{quote}
	\vspace{-5pt}
	Finally, the CMSA algorithm merges the three updated sets $B_i$, and the union $\bigcup B_i = \{b,e,l,s\}$ forms the minimal collapsible set containing $\{e, s\}$. The induced subgraph is shown in Fig.~\ref{Fig1}(c).
	
\end{example}


The next two theorems establish the correctness and computational complexity of Algorithm~\ref{alg-1}.

\begin{theorem}\label{thm-1}
The subset $B$ obtained by CMSA is the minimal collapsible set containing $A$.
\end{theorem}

\begin{proof}
Let $B_0$ denote the unique minimal collapsible set containing $A$, whose existence and uniqueness are guaranteed by Theorem~2.4 of \cite{wgh2011}. We aim to show that the subset $B$ obtained via CMSA satisfies $B = B_0$. Since the CMSA algorithm absorbs only close minimal separators, and by  {Lemma~\ref{prop-3}}, the set $B_0$ already contains all minimal separators between non-adjacent vertices in \(A\), it follows from the iterative construction that \(B \subseteq B_0\) holds trivially. Therefore, it suffices to prove that the graphical model \(\mathcal{L}\) is collapsible onto \(B\). 

Toward a contradiction, assume that $\mathcal{L}$ is not collapsible onto $B$. Then there exists a connected component $M$ of $G_{V \backslash B}$ and non-adjacent vertices $x, y \in N_G(M) \subseteq B$. Let $M_i$ denote the connected component of $G_{V\backslash A}$ containing $M$, and let $B_i$ be the subset returned from processing $M_i$. By construction, $M$ forms a connected component of $G_{A \cup M_i \backslash B_i}$, and its neighbourhood contains two non-adjacent vertices $x$ and $y$. This contradicts the termination condition of the CMSA algorithm. Hence, $\mathcal{L}$ is collapsible onto $B$, and we conclude that $B = B_0$, completing the proof.
\end{proof}

\begin{theorem}\label{thm-2}
The CMSA algorithm has a time complexity of $O(nm)$ and a space complexity of at most $O(n)$, where $n$ and $m$ denote the number of vertices and edges in the graph, respectively.
\end{theorem}
\begin{proof}
When the graph is represented in adjacency-list form, each close minimal separator can be identified in $O(m)$ time and $O(n)$ space complexity \citep{Takata2010}, and at most $n-2$ absorptions are required. The overall complexity is $O(nm)$ in time and $O(n)$ in space.
\end{proof}
\begin{remark}
Existing algorithms \citep{wgh2011,heng2023} for general graphs have $O(nm)$ time complexity and $O(n+m)$ space complexity.

\end{remark}

\vspace{-11pt}
\section{Experimental studies} \label{sec4}


\subsection{Efficiency of identifying minimal collapsible sets in general graphical models \label{sec-4.1}}

We evaluated the proposed algorithm through experiments. 
All algorithms were implemented in C and interfaced through Python. The complete source code is publicly accessible at \href{https://github.com/Balance-H/Algorithms}{https://github.com/Balance-H/Algorithms}. All experiments were conducted on a system equipped with an Intel Xeon Silver 4215R CPU and 128 GB of RAM.

We first evaluate the performance of the CMSA algorithm on general graphical models. Since SAHR is limited to decomposable graphs, we adopt the Induced Path Absorption (IPA) algorithm \citep{heng2023} as a baseline for comparison, as it has been shown to be the most effective method in \cite{heng2023}. For each combination of graph size $n \in \{2500, 5000, 7500, 10000\}$ and edge probability $p \in \{0.1, 0.01, 0.005, 0.001\}$, we generate 100 random trees with $n$ vertices and independently add edges with probability $p$ to obtain the corresponding random graphs. In each graph, we randomly select 10 target vertices and apply the corresponding algorithm to identify the minimal collapsible sets containing these vertices. We record the runtimes and compute the average over the 100 graphs for each $(n,p)$ configuration.

\begin{figure}[h]
	\centering
	\setlength{\belowcaptionskip}{-10pt} 
	\includegraphics[width=0.95\textwidth]{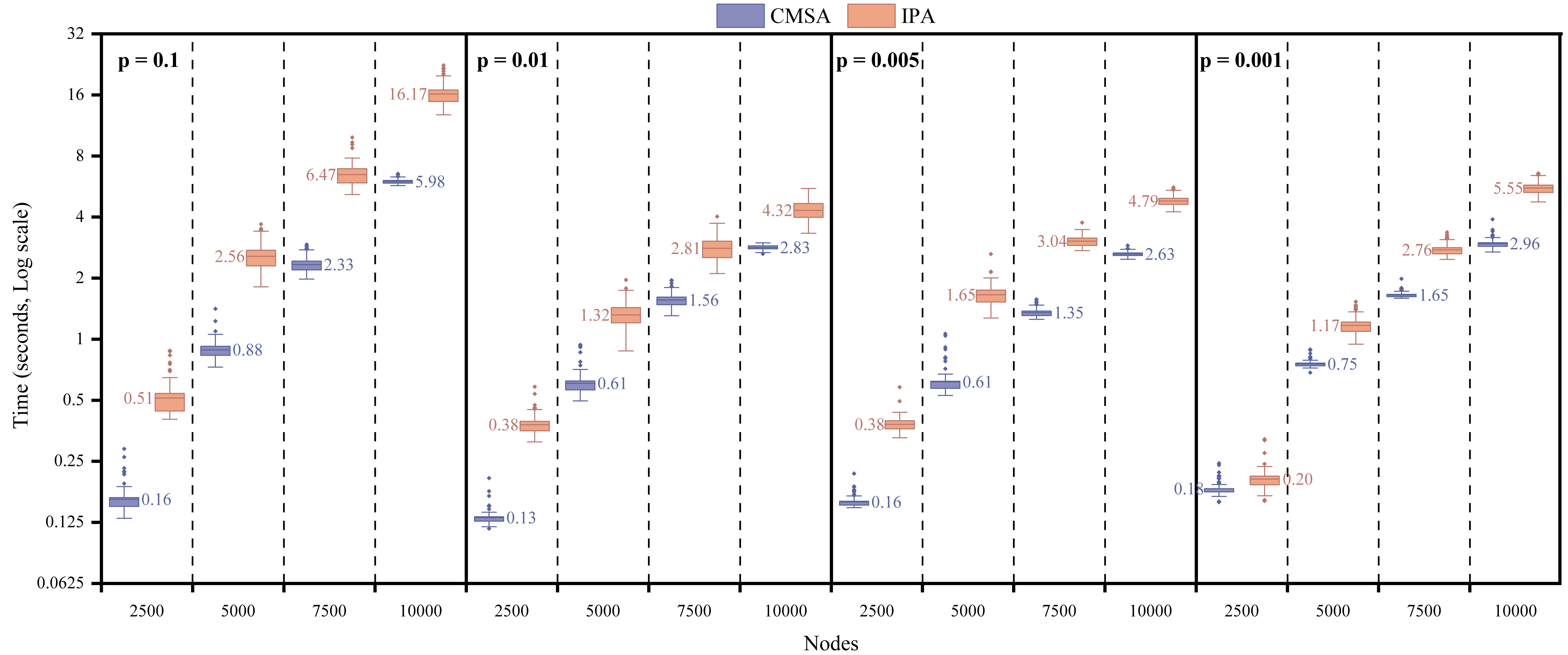} 
	\caption{Average runtimes (log-scaled, in seconds) of CMSA and IPA on networks of varying sizes.}
	\label{fig:example1}
\end{figure}

The experimental results are presented in Fig.~\ref{fig:example1}. The horizontal axis represents the number of nodes $n$ (separated by dashed lines), and the vertical axis shows the runtime in seconds on a logarithmic scale. The four subplots correspond to different edge densities $p$ (separated by solid lines). The blue boxplots represent the CMSA algorithm, and the orange ones represent the IPA algorithm, with the average runtime annotated beside each box. The results show that CMSA is substantially more efficient than IPA on dense graphs, with the performance advantage increasing as the number of nodes grows. As the graphs become sparser, the runtime of both algorithms decreases significantly, while CMSA consistently maintains superior efficiency.

\subsection{Efficiency of identifying minimal collapsible sets in decomposable graphical models \label{sec-4.2}}

We next evaluate the efficiency of CMSA and SAHR (\cite{madigan1990}) in identifying minimal collapsible sets within decomposable graphical models. For each combination of graph size $n \in \{250, 500, 750, 1000\}$ and edge probability $p \in \{0.1, 0.01\}$, we generate 100 random chordal graphs with approximately $n - 1 + 0.5 \, n(n - 1)p$ edges using a growing-subtree construction \citep{cseker2022generation}, where sparsity is controlled by the maximum subtree size. In each graph, we randomly select 10 target vertices and apply both CMSA and SAHR to identify the minimal collapsible sets containing these vertices. We record the runtimes and compute the mean runtime over the 100 graphs for each $(n,p)$ configuration.

\begin{table}[h]
	\caption{Average running times (s) of CMSA and SAHR in decomposable graphical models}
	\centering
	\resizebox{0.85\columnwidth}{!}{\begin{tabular}{lcccccccc}
			Nodes & \multicolumn{2}{c}{250} & \multicolumn{2}{c}{500}  & \multicolumn{2}{c}{750} & \multicolumn{2}{c}{1000} \\
			Average edges & 529 & 3334 & 1812 & 12912 & 3567 & 28652 & 6062 & 52959 \\
			CMSA & 0.0007 & 0.0012 & 0.0021 & 0.0047 & 0.0044 & 0.0112 & 0.0072 & 0.0248 \\
			SAHR & 0.0113 & 0.0611 & 0.0681 & 0.5455 & 0.1876 & 2.1648 & 0.3808 & 6.6983 \\
		\end{tabular}
	}
	\label{tab-1}
\end{table}

The experimental results in Table~\ref{tab-1} indicate that CMSA consistently outperforms SAHR across all graph sizes and densities. The advantage of CMSA becomes increasingly pronounced as the number of nodes and edges grows. This efficiency gain is primarily due to CMSA’s localized separator absorption strategy, which avoids the repeated boundary checks and global updates that dominate SAHR’s runtime in high-dimensional graphs.

\vspace{-11pt}

\section{Discussion}\label{sec5}

This note revisits a line of research tracing back to seminal works: \cite{asmussen1983} first formalized collapsibility in log-linear models, and \cite{madigan1990} introduced an elegant algorithm for finding minimal collapsible sets of decomposable graphs. We extend this line of work by showing that collapsibility admits a simple separator-based characterization: a model is collapsible onto a target set if and only if the set contains  {at least one minimal separator} between its non-adjacent vertices.
On this basis, we developed the Close Minimal Separator Absorption (CMSA) algorithm, which consistently outperforms existing methods and scales to high-dimensional graphs. Our findings highlight both theoretical and practical implications. Theoretically, they deepen the connection between collapsibility and local separator structure. Practically, they provide a tool that makes collapsibility analysis tractable in modern large-scale applications. We anticipate that the separator-based view may prove useful not only for contingency tables but also for broader classes of graphical models.

\vspace{-11pt}
\section*{Acknowledgments}
The authors are grateful to the Editor, Associate Editor, and anonymous reviewers for their constructive advice to improve the manuscript. This work was partially supported. Jianhua Guo and Pei Heng received support from the National Natural Science Foundation of China (Key Program No. 12431009) and the National Key Research and Development Program of China (Nos. 2020YFA0714100, 2020YFA0714102); Shiyuan He received support from the National Natural Science Foundation of China (No. 12571278); and Yi Sun received support from the National Natural Science Foundation of China Mathematical Tianyuan Fund (Nos. 12426105, 12426520). The first two authors contribute equally to this work.

\vspace{-11pt}
\section*{Supplementary Material}
The Supplementary Material includes an additional example illustrating the iterative procedures of SAHR and CMSA.

\bibliographystyle{biometrika}
\bibliography{ref} 

@article{asmussen1983,
  title={Collapsibility and response variables in contingency tables},
  author={Asmussen, S{\o}ren and Edwards, David},
  journal={Biometrika},
  volume={70},
  number={3},
  pages={567--578},
  year={1983},
  publisher={Oxford University Press}
}

@article{frydenberg1990,
  title   = {Marginalization and collapsibility in graphical interaction models},
  author  = {Frydenberg, Morten},
  journal = {The Annals of Statistics},
  volume  = {18},
  pages   = {790--805},
  year    = {1990}
}

@article{heng2023,
  title={Algorithms for convex hull finding in undirected graphical models},
  author={Heng, Pei and Sun, Yi},
  journal={Applied Mathematics and Computation},
  volume={445},
  pages={127852},
  year={2023},
  publisher={Elsevier}
}

@article{lauritzen1988local,
  author  = {Lauritzen, Steffen L. and Spiegelhalter, David J.},
  title   = {Local computations with probabilities on graphical structures and their application to expert systems (with discussion)},
  journal = {Journal of the Royal Statistical Society: Series B (Methodological)},
  year    = {1988},
  volume  = {50},
  number  = {2},
  pages   = {157--224},
  publisher = {Oxford University Press}
}

@article{madigan1990,
  title={An extension of the results of {Asmussen} and {Edwards} on collapsibility in contingency tables},
  author={Madigan, David and Mosurski, Krzysztof},
  journal={Biometrika},
  volume={77},
  number={2},
  pages={315--319},
  year={1990},
  publisher={Oxford University Press}
}

@article{cseker2022generation,
  title={Generation of random chordal graphs using subtrees of a tree},
  author={{\c{S}}eker, Oylum and Heggernes, Pinar and Ekim, Tinaz and Ta{\c{s}}k{\i}n, Z Caner},
  journal={RAIRO-Operations Research},
  volume={56},
  number={2},
  pages={565--582},
  year={2022},
  publisher={EDP Sciences}
}

@article{simpson1951interpretation,
  author    = {Edward H. Simpson},
  title     = {The Interpretation of Interaction in Contingency Tables},
  journal   = {Journal of the Royal Statistical Society: Series B (Methodological)},
  volume    = {13},
  number    = {2},
  pages     = {238--241},
  year      = {1951}
}

@article{Takata2010,
  title={Space-optimal, backtracking algorithms to list the minimal vertex separators of a graph},
  author={Takata, Ken},
  journal={Discrete Applied Mathematics},
  volume={158},
  number={15},
  pages={1660--1667},
  year={2010},
  publisher={Elsevier}
}

@article{wgh2011,
  title={Finding the minimal set for collapsible graphical models},
  author={Wang, Xiaofei and Guo, Jianhua and He, Xuming},
  journal={Proceedings of the American Mathematical Society},
  volume={139},
  number={1},
  pages={361--373},
  year={2011}
}

@article{yule1903notes,
  author    = {G. Udny Yule},
  title     = {Notes on the Theory of Association of Attributes in Statistics},
  journal   = {Biometrika},
  volume    = {2},
  number    = {2},
  pages     = {121--134},
  year      = {1903}
}

\newpage

\section{Supplementary Example}

In this supplementary section, we discuss in more depth about our proposed CMSA algorithm and the SAHR algorithm \citep{madigan1990}. Using the chordal graph shown in Fig.~\ref{Fig1}(a), the following two examples illustrate how the two algorithms identify the unique minimal collapsible set $\{b,e,l,s\}$ containing $A=\{e,s\}$. 

In Example~\ref{eg:supp:1}, we demonstrate that SAHR repeatedly examines all vertices in $V\setminus A$  to remove simplicial vertices. In general, the number of required iterations grows proportionally with the vertex size $|V\setminus A|$, and the computational cost of each iteration is highly sensitive to the ordering of vertices in $V\setminus A$, making SAHR difficult to scale to high-dimensional graphs.

\begin{figure}[h]
	\centering
	\begin{subfigure}[b]{0.19\textwidth}
		\centering
		\begin{tikzpicture}[
			scale=0.7,
			transform shape,
			every node/.style={circle, draw=black, minimum size=7mm, inner sep=0pt, font=\large}
			]
			\node (1) at (0,0){$a$};
			\node (2) at (0,-1){$t$};
			\node (3) at (1.8,-1){$l$};
			\node [fill=gray!30] (4) at (0.8,-2){$e$};
			\node (5) at (0,-3){$x$};
			\node (6) at (2.4,-3){$d$};
			\node (7) at (3,-1.3){$b$};
			\node [fill=gray!30] (8) at (2.4,0){$s$};
			\draw[-,thick] (2)--(1); \draw[-,thick] (2)--(3); \draw[-,thick] (4)--(3);
			\draw[-,thick] (2)--(4); \draw[-,thick] (4)--(5); \draw[-,thick] (4)--(6);
			\draw[-,thick] (4)--(7); \draw[-,thick] (8)--(3); \draw[-,thick] (8)--(7);
			\draw[-,thick] (6)--(7); \draw[-,thick] (3)--(7);
		\end{tikzpicture}
		\caption{} 
	\end{subfigure}
	\hfill 
	\begin{subfigure}[b]{0.19\textwidth}
		\centering
		\begin{tikzpicture}[
			scale=0.7,
			transform shape,
			every node/.style={circle, draw=black, minimum size=7mm, inner sep=0pt, font=\large}
			]
			\node (2) at (0,-1){$t$};
			\node (3) at (1.8,-1){$l$};
			\node [fill=gray!30] (4) at (0.8,-2){$e$};
			\node (5) at (0,-3){$x$};
			\node (6) at (2.4,-3){$d$};
			\node (7) at (3,-1.3){$b$};
			\node [fill=gray!30] (8) at (2.4,0){$s$};
			\draw[-,thick] (2)--(3); \draw[-,thick] (4)--(3);
			\draw[-,thick] (2)--(4); \draw[-,thick] (4)--(5); \draw[-,thick] (4)--(6);
			\draw[-,thick] (4)--(7); \draw[-,thick] (8)--(3); \draw[-,thick] (8)--(7);
			\draw[-,thick] (6)--(7); \draw[-,thick] (3)--(7);
		\end{tikzpicture}
		\caption{} 
	\end{subfigure}
	\hfill
	\begin{subfigure}[b]{0.19\textwidth}
		\centering
		\begin{tikzpicture}[
			scale=0.7,
			transform shape,
			every node/.style={circle, draw=black, minimum size=7mm, inner sep=0pt, font=\large}
			]
			\node (3) at (1.8,-1){$l$};
			\node [fill=gray!30] (4) at (0.8,-2){$e$};
			\node (5) at (0,-3){$x$};
			\node (6) at (2.4,-3){$d$};
			\node (7) at (3,-1.3){$b$};
			\node [fill=gray!30] (8) at (2.4,0){$s$};
			\draw[-,thick] (4)--(3);
			\draw[-,thick] (4)--(5); \draw[-,thick] (4)--(6);
			\draw[-,thick] (4)--(7); \draw[-,thick] (8)--(3); \draw[-,thick] (8)--(7);
			\draw[-,thick] (6)--(7);\draw[-,thick] (3)--(7);
		\end{tikzpicture}
		\caption{} 
	\end{subfigure}
	\hfill 
	\begin{subfigure}[b]{0.19\textwidth}
		\centering
		\begin{tikzpicture}[
			scale=0.7,
			transform shape,
			every node/.style={circle, draw=black, minimum size=7mm, inner sep=0pt, font=\large}
			]
			\node (3) at (1.8,-1){$l$};
			\node [fill=gray!30] (4) at (0.8,-2){$e$};
			\node (6) at (2.4,-3){$d$};
			\node (7) at (3,-1.3){$b$};
			\node [fill=gray!30] (8) at (2.4,0){$s$};
			\draw[-,thick] (4)--(3);
			\draw[-,thick] (4)--(6);
			\draw[-,thick] (4)--(7); \draw[-,thick] (8)--(3); \draw[-,thick] (8)--(7);
			\draw[-,thick] (6)--(7);\draw[-,thick] (3)--(7);
		\end{tikzpicture}
		\caption{} 
	\end{subfigure}
	\hfill 
	\begin{subfigure}[b]{0.19\textwidth}
		\centering
		\begin{tikzpicture}[
			scale=0.7,
			transform shape,
			every node/.style={circle, draw=black, minimum size=7mm, inner sep=0pt, font=\large}
			]
			\node (3) at (1.8,-1){$l$};
			\node [fill=gray!30] (4) at (0.8,-2){$e$};
			\node (7) at (3,-1.3){$b$};
			\node [fill=gray!30] (8) at (2.4,0){$s$};
			\draw[-,thick] (4)--(3);\draw[-,thick] (3)--(7);
			\draw[-,thick] (4)--(7); \draw[-,thick] (8)--(3); \draw[-,thick] (8)--(7);
		\end{tikzpicture}
		\caption{} 
	\end{subfigure}
	\caption{\textbf{(a)}. The chordal graph obtained from the moral graph of the Asia network \citep{lauritzen1988local} by adding the edge $(b,l)$; \textbf{(b)}. The graph obtained from (a) by removing the simplicial vertex $a$; \textbf{(c)}. The graph obtained from (b) by removing the simplicial vertex $t$; \textbf{(d)}. The graph obtained from (c) by removing the simplicial vertex $x$; \textbf{(e)}. The graph obtained from (d) by removing the simplicial vertex $d$.}
	\label{Fig1}
\end{figure}

\begin{example} \label{eg:supp:1}
	We first present in Fig.~\ref{Fig1} a complete illustration of how the SAHR algorithm \citep{madigan1990} iteratively removes simplicial vertices to identify the minimal collapsible set containing the target variable set \( \{e, s\} \).
	Since the SAHR algorithm does not specify a fixed elimination order, at each iteration it traverses all vertices in \( V \setminus \{e, s\} \) until a simplicial vertex is identified and removed, and then repeats the same procedure.
	As a result, its efficiency depends heavily on the ordering of vertices in \( V \setminus \{e, s\} \).
	We illustrate this point through the two cases below. \\
	
	\noindent \textbf{Case 1.} Suppose the vertices in \( V \setminus \{e, s\} \) are ordered as
	\( \{a, t, x, d, l, b\} \),
	and the algorithm examines vertices according to this order at each iteration.
	\begin{itemize}  \setlength{\itemsep}{0pt}
		\item[Step 1.]
		The algorithm checks vertex \( a \) and identifies it as simplicial. Vertex \( a \) is removed.
		\item[Step 2.]
		For the remaining vertices \( \{t, x, d, l, b\} \) in Fig.~\ref{Fig1}(b), the algorithm checks vertex \( t \), finds it simplicial, and removes it.
		\item[Step 3.]
		For the remaining vertices \( \{x, d, l, b\} \) in Fig.~\ref{Fig1}(c), the algorithm checks vertex \( x \), finds it simplicial, and removes it.
		\item[Step 4.]
		For the remaining vertices \( \{d, l, b\} \) in Fig.~\ref{Fig1}(d), the algorithm checks vertex \( d \), finds it simplicial, and removes it.
		\item[Step 5.]
		For the remaining vertices \( \{l, b\} \) in Fig.~\ref{Fig1}(e), the algorithm finds that neither \( l \) nor \( b \) is simplicial and terminates.
	\end{itemize}
	
	\vspace{10pt}
	
	\noindent\textbf{Case 2.} Suppose instead that the vertices in \( V \setminus \{e, s\} \) are ordered as
	\( \{l, b, t, a, x, d\} \),
	and the algorithm again examines vertices according to this order at each iteration.
	\begin{itemize} \setlength{\itemsep}{0pt}
		\item[Step 1.]
		The algorithm checks vertices \( l \), \( b \), and \( t \) and finds none of them simplicial.
		It then checks vertex \( a \), identifies it as simplicial, and removes it.
		\item[Step 2.]
		For the remaining vertices \( \{l, b, t, x, d\} \) in Fig.~\ref{Fig1}(b), the algorithm checks \( l \) and \( b \) and finds them not simplicial.
		It then checks vertex \( t \), finds it simplicial, and removes it.
		\item[Step 3.]
		For the remaining vertices \( \{l, b, x, d\} \) in Fig.~\ref{Fig1}(c), the algorithm again checks \( l \) and \( b \) without success,
		then checks vertex \( x \), finds it simplicial, and removes it.
		\item[Step 4.]
		For the remaining vertices \( \{l, b, d\} \) in Fig.~\ref{Fig1}(d), the algorithm checks \( l \) and \( b \), finds them not simplicial,
		then checks vertex \( d \), finds it simplicial, and removes it.
		\item[Step 5.]
		For the remaining vertices \( \{l, b\} \) in Fig.~\ref{Fig1}(e), the algorithm finds that neither vertex is simplicial and terminates.
	\end{itemize}
	
	In each iteration, the algorithm must repeatedly examine the remaining vertices to identify a simplicial one.
	This is because the simpliciality of a vertex may change as other vertices are removed, as illustrated by vertex \( t \) in Case~2.
	In general, especially for large graphs, an optimal vertex ordering (as in Case~1) cannot be known in advance.
	Such redundant checks, as seen in Case~2, can substantially degrade performance, particularly in high-dimensional graphs.
\end{example}

In contrast, the proposed CMSA algorithm exploits a local structure searching strategy, which significantly reduces the computational cost for general high-dimensional graphs.  Specifically, “locality” refers to the way the proposed algorithm searches for close minimal separators within node neighbourhoods. During the iterations, the search is restricted to several connected components and their associated neighbourhoods, treated separately. More importantly, the size of each connected component decreases progressively, leading to a corresponding reduction in the search effort. We illustrate this by Example~\ref{eg:supp:2} below.

\begin{example} \label{eg:supp:2}
	We illustrate in Fig.~\ref{Fig2} the procedure of the proposed CMSA algorithm for iteratively absorbing minimal separators to identify the minimal collapsible set. The steps are as follows.
	
	\begin{figure}[t]
		\centering
		\begin{subfigure}[b]{0.19\textwidth}
			\centering
			\begin{tikzpicture}[
				scale=0.6,
				transform shape,
				every node/.style={circle, draw=black, minimum size=7mm, inner sep=0pt, font=\large}
				]
				\node (1) at (0,0){$a$};
				\node (2) at (0,-1){$t$};
				\node (3) at (1.8,-1){$l$};
				\node [fill=gray!30] (4) at (0.8,-2){$e$};
				\node (5) at (0,-3){$x$};
				\node (6) at (2.4,-3){$d$};
				\node (7) at (3,-1.3){$b$};
				\node [fill=gray!30] (8) at (2.4,0){$s$};
				\draw[-,thick] (2)--(1); \draw[-,thick] (2)--(3); \draw[-,thick] (4)--(3);
				\draw[-,thick] (2)--(4); \draw[-,thick] (4)--(5); \draw[-,thick] (4)--(6);
				\draw[-,thick] (4)--(7); \draw[-,thick] (8)--(3); \draw[-,thick] (8)--(7);
				\draw[-,thick] (6)--(7); \draw[-,thick] (3)--(7);
			\end{tikzpicture}
			\caption{} 
		\end{subfigure}
		\hfill 
		\begin{subfigure}[b]{0.19\textwidth}
			\centering
			\begin{tikzpicture}[
				scale=0.6,
				transform shape,
				every node/.style={circle, draw=black, minimum size=7mm, inner sep=0pt, font=\large}
				]
				\node (1) at (0,0){$a$};
				\node (2) at (0,-1){$t$};
				\node (3) at (1.8,-1){$l$};
				\node (5) at (0,-3){$x$};
				\node (6) at (2.4,-3){$d$};
				\node (7) at (3,-1.3){$b$};
				\draw[-,thick] (2)--(3); \draw[-,thick] (2)--(1);
				\draw[-,thick] (6)--(7); \draw[-,thick] (3)--(7);
			\end{tikzpicture}
			\caption{} 
		\end{subfigure}
		\hfill
		\begin{subfigure}[b]{0.19\textwidth}
			\centering
			\begin{tikzpicture}[
				scale=0.6,
				transform shape,
				every node/.style={circle, draw=black, minimum size=7mm, inner sep=0pt, font=\large}
				]
				\node (1) at (0,0){$a$};
				\node (2) at (0,-1){$t$};
				\node (3) at (1.8,-1){$l$};
				\node [fill=gray!30] (4) at (0.8,-2){$e$};
				\node (6) at (2.4,-3){$d$};
				\node (7) at (3,-1.3){$b$};
				\node [fill=gray!30] (8) at (2.4,0){$s$};
				\draw[-,thick] (2)--(1); \draw[-,thick] (2)--(3); \draw[-,thick] (4)--(3);
				\draw[-,thick] (2)--(4); \draw[-,thick] (4)--(6);
				\draw[-,thick] (4)--(7); \draw[-,thick] (8)--(3); \draw[-,thick] (8)--(7);
				\draw[-,thick] (6)--(7); \draw[-,thick] (3)--(7);
			\end{tikzpicture}
			\caption{} 
		\end{subfigure}
		\hfill 
		\begin{subfigure}[b]{0.19\textwidth}
			\centering
			\begin{tikzpicture}[
				scale=0.6,
				transform shape,
				every node/.style={circle, draw=black, minimum size=7mm, inner sep=0pt, font=\large}
				]
				\node (1) at (0,0){$a$};
				\node (2) at (0,-1){$t$};
				\node [fill=gray!30] (3) at (1.8,-1){$l$};
				\node [fill=gray!30] (4) at (0.8,-2){$e$};
				\node (6) at (2.4,-3){$d$};
				\node [fill=gray!30] (7) at (3,-1.3){$b$};
				\node [fill=gray!30] (8) at (2.4,0){$s$};
				\draw[-,thick] (2)--(1); \draw[-,thick] (2)--(3); \draw[-,thick] (4)--(3);
				\draw[-,thick] (2)--(4); \draw[-,thick] (4)--(6);
				\draw[-,thick] (4)--(7); \draw[-,thick] (8)--(3); \draw[-,thick] (8)--(7);
				\draw[-,thick] (6)--(7); \draw[-,thick] (3)--(7);
			\end{tikzpicture}
			\caption{} 
		\end{subfigure}
		\hfill
		\begin{subfigure}[b]{0.19\textwidth}
			\centering
			\begin{tikzpicture}[
				scale=0.6,
				transform shape,
				every node/.style={circle, draw=black, minimum size=7mm, inner sep=0pt, font=\large}
				]
				\node (1) at (0,0){$a$};
				\node (2) at (0,-1){$t$};
				\node (6) at (2.4,-3){$d$};
				
				\draw[-,thick] (1)--(2);
			\end{tikzpicture}
			\caption{} 
		\end{subfigure}
		
		\caption{\textbf{(a)}. The chordal graph obtained from the moral graph of the Asia network \citep{lauritzen1988local} by adding the edge $(b,l)$; 
			\textbf{(b)}. The graph obtained from (a) by removing the variables of interest $A = \{e, s\}$; 
			\textbf{(c)}. The graph obtained by combining the connected component $\{a, t, l, b, d\}$ with the target variables $A = \{e, s\}$;
			\textbf{(d)}. The graph obtained after one absorption operation, with the variables of interest updated to $A = \{b, e, l, s\}$; 
			\textbf{(e)}. The graph obtained from (d) by removing the variables of interest $A = \{b, e, l, s\}$.}
		
		\label{Fig2}
	\end{figure}
	
	\begin{enumerate} \setlength{\itemsep}{0pt}
		\item[Step 1.]
		The algorithm first identifies the connected components of \( V \setminus \{e, s\} \).
		As shown in Fig.~\ref{Fig2}(b), there are two components:
		\[
		M_1 = \{a, t, l, b, d\}, 
		\qquad
		M_2 = \{x\}.
		\]
		The component \( M_2 \) has a neighbourhood consisting of a single vertex and therefore requires no further operation.
		
		\item[Step 2.]
		For the component \( M_1 \), we combine it with the target variable set \( \{e, s\} \) to obtain the graph shown in Fig.~\ref{Fig2}(c).
		We then apply the \textsc{CloseSeparator} algorithm \citep{Takata2010} to search separately for minimal \( e\!s \)-separators close to \( e \) and close to \( s \) in Fig.~\ref{Fig2}(c).
		In this case, both searches yield the same set \( \{l, b\} \).
		Absorbing \( \{l, b\} \) into the target set \( \{e, s\} \) yields the updated graph structure shown in Fig.~\ref{Fig2}(d).
		
		\item[Step 3.]
		In the second iteration, CMSA identifies the connected components of
		\( V \setminus \{b, e, l, s\} \).
		As shown in Fig.~\ref{Fig2}(e), there are two components:
		\[
		M_3 = \{a, t\}, 
		\qquad
		M_4 = \{d\}.
		\]
		Both components \( M_3 \) and \( M_4 \) in Fig.~\ref{Fig2}(d) have complete neighbourhood sets.
		Therefore, the algorithm terminates and outputs the minimal collapsible set
		\( \{b, e, l, s\} \).
	\end{enumerate}
\end{example}

In summary, compared with CMSA,  SAHR generally requires more iterations; and each iteration can also be more complex, depending on the graph and the vertex order for simpliciality checks. On the other hand, CMSA directly identifies and absorbs close minimal separators within local neighbourhoods. CMSA typically converges in far fewer iterations. As a result, to find a minimal collapsible set containing a target set $A$, our computational savings arise from the reduced number of iterations and the avoidance of repeated scan within the vertices of $V\setminus A$.
\end{document}